\documentclass[11pt,reqno]{amsart}
\usepackage{amsfonts}
\usepackage{mathrsfs}
\allowdisplaybreaks[4]
\usepackage{graphicx} %We can use any other package if it is necessary
\usepackage{epsfig}
\usepackage{amssymb}
\usepackage{amsmath}
\usepackage{cite}
\newtheorem{theorem}{Theorem}

\newtheorem{proposition}[theorem]{Proposition}

\newtheorem{lemma}[theorem]{Lemma}
\newcommand{\be}{\begin{equation}}
\newcommand{\ee}{\end{equation}}
\newcommand{\bea}{\begin{eqnarray}}
\newcommand{\eea}{\end{eqnarray}}
\newcommand{\ba}{\begin{array}}
\newcommand{\ea}{\end{array}}
\newcommand{\bean}{\begin{eqnarray*}}
\newcommand{\eean}{\end{eqnarray*}}

\newcommand{\pa}{\partial}

\begin{document}
\title{the squared eigenfunction symmetries for the BTL and CTL Hierarchies}
\author{Jipeng Cheng$^{1}$, Jingsong He$^{2*}$
 }
\dedicatory { 1\ Department of Mathematics, China University of
Mining and Technology, Xuzhou, Jiangsu 221116 , P.\ R.\ China\\
2\ Department of Mathematics, Ningbo University,
Ningbo, Zhejiang 315211, P.\ R.\ China \\
}

\thanks{$^*$Corresponding author. email:hejingsong@nbu.edu.cn; jshe@ustc.edu.cn.}
\begin{abstract}
In this paper, the squared eigenfunction symmetries for the BTL and
CTL hierarchies are explicitly constructed with the suitable
modification of the ones for the TL hierarchy, by considering the
BTL and CTL constraints. Also the connections with the corresponding
additional symmetries are investigated: the squared eigenfunction
symmetry generated by the wave function can be viewed as the
generating function for the additional symmetries.



\end{abstract}
\maketitle
%%%%%%%%%%%%%%%%%%%%%%%%%%%%%%%%%%%%%%%%%%%%%%%%%%%%%%
\section{Introduction}
%%%%%%%%%%%%%%%%%%%%%%%%%%%%%%%%%%%%%%%%%%%%%%%%%%%%%%%
The Toda lattice (TL) equation\cite{toda}, as an important
integrable system, describes the motion of one-dimensional particles
with exponential interaction of neighbors, which plays significant
role in physics. The TL hierarchy, which is one of the most
important integrable hierarchies, was first introduced by Ueno and
Takasaki in \cite{uenotaksasai} to generalize the Toda lattice
equations\cite{toda} along the work of the KP hierarchy\cite{DJKM}.
In \cite{uenotaksasai}, the analogues of the B and C types for the
TL hierarchy, i.e. the BTL and CTL hierarchies, are also considered,
which are corresponding to infinite dimensional Lie algebras
$\textmd{o}(\infty)$ and $\textmd{sp}(\infty)$ respectively. The BTL
and CTL hierarchies are also very important in the integrable system
just like the TL hierarchy\cite{uenotaksasai}. However, there are
few researches on the BTL and CTL hierarchies in literature. So much
work can be done for the BTL and CTL hierarchies.

The squared eigenfunction
symmetry\cite{oevel1993,oevel1994,1OC98,2OC98}, also called ``ghost"
symmetry\cite{aratyn1998}, is a kind of symmetry generated by
eigenfunctions and adjoint eigenfunctions in the integrable system.
The squared eigenfunction symmetry has many applications in the
integrable system. For example, 1) symmetry
constraint\cite{oevel1994,2OC98,Cheng91,SS91,KS92,Cheng92,LorisWillox99,Tu2011}
can be defined by identifying the squared eigenfunction symmetry
with the usual flow of the integrable hierarchy; 2) the connection
with the additional
symmetry\cite{aratyn1998,Cheng10,Li12,cheng2012}, which is the
symmetry depending explicitly on the space and time
variables\cite{Fokas1981,Chen1983,OS86,ASM95,D95,takasaki1996,Tu07,he2007,tian2011};
3) the extended integrable systems\cite{1Zeng08, 2Zeng08}, which
contain the integrable equations with self-consistent sources, can
be constructed with the help of the squared eigenfunction symmetry.
Recently, the squared eigenfunction symmetries for the BKP hierarchy
and the discrete KP hierarchy are systematically developed in
\cite{Cheng10} and \cite{Li12} respectively. Also the squared
eigenfunction symmetry for the TL hierarchy and its connection with
the additional symmetry are investigated in \cite{cheng2012}. In
this paper, we will concentrate on the construction of the squared
eigenfunction symmetry of the BTL and CTL hierarchies.

The squared eigenfunction symmetry of the Toda lattice hierarchy
\cite{cheng2012} is given in the form of the Kronecker product of
the vector eigenfunctions and the vector adjoint eigenfunctions.
Because of the BTL and CTL constraints, the squared eigenfunction
symmetry can not be defined directly from the one of the TL
hierarchy and some modification must be needed. For this, we
construct the squared eigenfunction symmetries of the BTL and CTL
hierarchies by the suitable combination of the ones for the TL
hierarchy. Then the connection with the additional symmetry is
investigated: the particular squared eigenfunction symmetries
generated by the wave functions can be viewed as the generating
functions of the additional symmetries.

This paper is organized in the following way. In Section 2, we
recall some basic knowledge about the BTL and CTL hierarchies. Then,
we construct the squared eigenfunction symmetry for the BTL
hierarchy in Section 3. Next, in Section 4 the squared eigenfunction
symmetry for the CTL hierarchy is also investigated. At last, we
devote Section 5 to some conclusions and discussions.

%%%%%%%%%%%%%%%%%%%%%%%%%%%%%%%%%%%%%%%%%%%%%%%%%%%%%%
\section{the BTL and CTL hierarchies}
%%%%%%%%%%%%%%%%%%%%%%%%%%%%%%%%%%%%%%%%%%%%%%%%%%%%%%%
In this section, some basic facts about the BTL and CTL hierarchies
are reviewed. One can refer to \cite{uenotaksasai} for
 more details about the BTL and CTL hierarchies.

Firstly, consider the algebra
$$ \mathscr{D}=\{(P_1,P_2)\in\textmd{ gl}((\infty))\times \textmd{gl}((\infty))\ | \ (P_1)_{ij}=0 \ \textmd{ for}\  j-i\gg0, \ (P_2)_{ij}=0 \ \textmd{for} \  i-j\gg0\},$$
which has the following splitting:
\begin{eqnarray*}
\mathscr{D}&=&\mathscr{D}_+ +\mathscr{D}_-,\\
\mathscr{D}_+&=&\{(P,P)\in \mathscr{D}\ | \ (P)_{ij}=0 \ \textmd{for} \  |i-j|\gg 0\}=\{(P_1,P_2)\in \mathscr{D}\ | \ P_1=P_2\},\\
\mathscr{D}_-&=&\{(P_1,P_2)\in \mathscr{D}\ | \ (P_1)_{ij}=0\ \textmd{for}\ j\geq i, \ (P_2)_{ij}=0 \ \textmd{for} \  i> j\},
\end{eqnarray*}
with $(P_1,P_2)=(P_1,P_2)_+ +(P_1,P_2)_-$ given by
$$(P_1,P_2)_+=(P_{1u}+P_{2l},P_{1u}+P_{2l}),\ (P_1,P_2)_-=(P_{1l}-P_{2l},P_{2u}-P_{1u}),$$
where for a matrix $P$, $P_u$ and $P_l$ denote the upper (including diagonal) and strictly lower triangular parts of $P$, respectively. For $(P_1,P_2),(Q_1,Q_2)\in  \mathscr{D}$, we define $$(P_1,P_2)(Q_1,Q_2)=(P_1Q_1,P_2Q_2),\quad (P_1,P_2)^{-1}=(P_1^{-1},P_2^{-1}).$$

Then the BTL (or CTL) hierarchy is defined in the Lax forms as
\begin{equation}\label{bctlhierarchy}
    \pa_{x_{2n+1}}L=[(L^{2n+1}_1,0)_+,L]\ \ \ \textmd{and}\ \ \ \pa_{y_{2n+1}}L=[(0,L_2^{2n+1})_+,L],\ \ \ n=0,1,2,\cdots
\end{equation}
where the Lax operator $L$ is given by  a pair of infinite matrices
\begin{equation}\label{laxoperator}
    L=(L_1,L_2)=\Big(\sum _{-\infty<i\leq 1}\textmd{diag}[a_i^{(1)}(s)]\Lambda^i,\sum _{-1\leq i<\infty}\textmd{diag}[a_i^{(2)}(s)]\Lambda^i\Big)\in  \mathscr{D}
\end{equation}
with $\Lambda=(\delta_{j-i,1})_{i,j\in \mathbb{Z}}$, and
$a_i^{(k)}(s)$ and  $a_i^{(k)}(s)$ depending on
$x=(x_1,x_3,x_5,\cdots)$ and $y=(y_1,y_3,y_5,\cdots)$, such that
$$a_1^{(1)}(s)=1 \ \ \ \text{and}\ \ \ a_{-1}^{(2)}(s)\neq 0\ \ \ \forall s$$
and satisfies the BTL (or CTL) constraint\cite{uenotaksasai}
\begin{equation}\label{bctlconstr}
    L^T=-(J,J)L(J^{-1},J^{-1})\ \left(\textmd{or}\ L^T=-(K,K)L(K^{-1},K^{-1})\right),
\end{equation}
where $J=((-1)^i\delta_{i+j,0})_{i,j\in\mathbb{Z}}$, $K=\Lambda J$
and $T$ refers to the matrix transpose. The BTL (or CTL) constraint
on the components of the Lax operators $L_1$ and $L_2$ is explicitly
showed as
\begin{equation}\label{bctlconstrcomponent}
    a_i^{(k)}(s)=(-1)^{i+1}a_i^{(k)}(-s-i)\ \left(\textmd{or}\
    a_i^{(k)}(s)=(-1)^{i+1}a_i^{(k)}(-s-i-1)\right), k=1,2.
\end{equation}

The Lax equation for the BTL (or CTL) hierarchy can be expressed as
a system of equations of the Zakharov-Shabat type:
\begin{eqnarray}
&&\pa_{x_{2n+1}}(L^{2m+1}_1)_u-\pa_{x_{2m+1}}(L^{2n+1}_1)_u+[(L^{2m+1}_1)_u,(L^{2n+1}_1)_u]=0,\label{zs1}\\
&&\pa_{y_{2n+1}}(L^{2m+1}_2)_l-\pa_{y_{2m+1}}(L^{2n+1}_2)_l+[(L^{2m+1}_2)_l,(L^{2n+1}_2)_l]=0,\label{zs2}\\
&&\pa_{y_{2n+1}}(L^{2m+1}_1)_u-\pa_{x_{2m+1}}(L^{2n+1}_2)_l+[(L^{2m+1}_1)_u,(L^{2n+1}_2)_l]=0,\
\ \ m,n=0,1,2,\cdots\label{zs3}
\end{eqnarray}
When $m=n=0$, one can from (\ref{zs3}) get the BTL equation
\begin{eqnarray}
&&\pa_{x_1}a_{-1}^{(2)}(1)=a_{-1}^{(2)}(1)a_{0}^{(1)}(1),\ \
\pa_{x_1}a_{-1}^{(2)}(s)=a_{-1}^{(2)}(s)(a_{0}^{(1)}(s)-a_{0}^{(1)}(s-1))\
(s\geq
2),\nonumber\\
&&\pa_{y_1}a_{0}^{(1)}(s)=a_{-1}^{(2)}(s)-a_{-1}^{(2)}(s+1)\
(s\geq1),\label{btlequation}
\end{eqnarray}
and the CTL equation
\begin{eqnarray}
&&\pa_{x_1}a_{-1}^{(2)}(0)=2a_{-1}^{(2)}(0)a_{0}^{(1)}(0),\ \
\pa_{x_1}a_{-1}^{(2)}(s)=a_{-1}^{(2)}(s)(a_{0}^{(1)}(s)-a_{0}^{(1)}(s-1))\
(s\geq
1),\nonumber\\
&&\pa_{y_1}a_{0}^{(1)}(s)=a_{-1}^{(2)}(s)-a_{-1}^{(2)}(s+1)\
(s\geq0),\label{ctlequation}
\end{eqnarray}
by considering the corresponding constraint
(\ref{bctlconstrcomponent}).

The Lax operator of the BTL (or CTL) hierarchy (\ref{bctlhierarchy})
has the representation
\begin{equation}\label{laxwaveexpression}
    L=W(\Lambda,\Lambda^{-1})W^{-1}=S(\Lambda,\Lambda^{-1})S^{-1}
\end{equation}
in terms of two pairs of wave operators $W=(W_1,W_2)$ and $S=(S_1,S_2)$, where
\begin{eqnarray}
S_1(x,y)=\sum_{i\geq 0} \textmd{diag}[c_i(s;x,y)] \Lambda^{-i},&&S_2(x,y)=\sum_{i\geq 0} \textmd{diag}[c_i'(s;x,y)] \Lambda^{i} \label{smatrices}
\end{eqnarray}
and
\begin{equation}\label{wmatrices}
    W_1(x,y)=S_1(x,y)e^{\xi(x,\Lambda)},\quad W_2(x,y)=S_2(x,y)e^{\xi(y,\Lambda^{-1})}
\end{equation}
with $c_0(s;x,y)=1$ and $c_o'(s;x,y)\neq 0$ for any $s$, and $\xi(x,\Lambda^{\pm1})=\sum_{n\geq0}x_{2n+1}\Lambda^{\pm 2n+1}$. Obviously,
$W=(W_1,W_2)$ are not uniquely determined, but have the arbitrariness
$$W_1(x,y)\mapsto W_1(x,y)f^{1}(\Lambda),\quad W_2(x,y)\mapsto W_2(x,y)f^{2}(\Lambda).$$
Here $f^{1}(\lambda)=\sum_{i\geq0}f_i^1\lambda^{-i}$ and $f^{2}(\lambda)=\sum_{i\geq0}f_i^2\lambda^{i}$ ($f_0^1=1$, $f_0^2\neq0$)
are formal Laurent series with constant scalar coefficients. Under an appropriate choice of $f_i(\lambda)$, $W=(W_1,W_2)$ satisfies
\begin{equation}\label{bcwconstraints}
    J^{-1}W_i^TJ=W_i^{-1}\ \textmd{for}\ \textmd{BTL}\ (\ \textmd{or}\ K^{-1}W_i^TK=W_i^{-1}\ \textmd{for}\ \textmd{CTL}), i=1,2.
\end{equation}
The wave operators evolve according to
\begin{eqnarray}
\pa_{x_{2n+1}}S&=&-(L_1^{2n+1},0)_-S,\quad \pa_{y_{2n+1}}S=-(0,L_2^{2n+1})_-S,\label{sevolution}\\
\pa_{x_{2n+1}}W&=&(L_1^{2n+1},0)_+W,\quad \pa_{y_{2n+1}}W=(0,L_2^{2n+1})_+W.\label{wevolution}
\end{eqnarray}

The vector wave functions $\Psi=(\Psi_1,\Psi_2)$ and the adjoint wave functions $\Psi^*=(\Psi_1^*,\Psi_2^*)$, can also be introduced as
\begin{eqnarray}
  \Psi_i(x,y;\lambda)&=& (\Psi_i(n;x,y;\lambda))_{n\in\mathbb{Z}}:=W_i(x,y)\chi(\lambda),\label{wavefunction} \\
  \Psi_i^*(x,y;\lambda)&=&(\Psi_i^*(n;x,y;\lambda))_{n\in\mathbb{Z}}:=(W_i(x,y)^{-1})^T\chi^*(\lambda),\label{adjointwavefunction}
\end{eqnarray}
with $\chi(\lambda)=(\lambda^i)_{i\in\mathbb{Z}}$ and
$\chi^*(\lambda)=\chi(\lambda^{-1})$, which satisfy the following
relations:
\begin{eqnarray}
L\Psi=(z,z^{-1})\Psi,&& L^T\Psi^*=(z,z^{-1})\Psi^*\label{laxactonwave}\\
\pa_{x_{2n+1}}\Psi=(L_1^{2n+1},0)_+\Psi,&& \pa_{y_{2n+1}}\Psi=(0,L_2^{2n+1})_+\Psi,\label{vectorwavefunctionequation}\\
\pa_{x_{2n+1}}\Psi^*=-(L_1^{2n+1},0)_+^T\Psi^*,&&
\pa_{y_{2n+1}}\Psi^*=-(0,L_2^{2n+1})_+^T\Psi^*.\label{vectoradjointwavefunctionequation}
\end{eqnarray}
From the BTL (or CTL) constraint (\ref{bcwconstraints}) on the wave
operators, the adjoint wave function is connected with the wave
function in the following way,
\begin{equation}\label{relationadjoint}
    \Psi_i^*(x,y,\lambda)=J\Psi_i(x,y,-\lambda)\ \ \ (\text{or}\ \Psi_i^*(x,y,\lambda)=\lambda
    K\Psi_i(x,y,-\lambda)).
\end{equation}

If vector functions $q=(q(n;x,y))_{n\in\mathbb{Z}}$ and
$r=(r(n;x,y))_{n\in\mathbb{Z}}$ satisfy
\begin{eqnarray}
\pa_{x_{2n+1}}q=(L_1^{2n+1})_uq,&&\pa_{y_{2n+1}}q=(L_2^{2n+1})_lq,\nonumber\\
\pa_{x_{2n+1}}r=-(L_1^{2n+1})^T_uq,&&\pa_{y_{2n+1}}r=-(L_2^{2n+1})^T_lr,\label{qrdef}
\end{eqnarray}
we call them \textbf{vector eigenfunction} and \textbf{vector
adjoint eigenfunction} for the BTL (or CTL) hierarchy respectively.
Obviously, the wave functions $\Psi_1$ and $\Psi_2$ are
eigenfunctions, and the adjoint wave functions $\Psi_1^*$ and
$\Psi_2^*$ are the adjoint eigenfunctions. From the BTL ( or CTL)
constraint (\ref{bctlconstr}), one can know that
\begin{eqnarray}
(L_1^{2n+1})_u^T=-J(L_1^{2n+1})_uJ^{-1},&&(L_2^{2n+1})_l^T=-J(L_2^{2n+1})_lJ^{-1}\nonumber\\
\Big(\text{or}\
(L_1^{2n+1})_u^T=-K(L_1^{2n+1})_uK^{-1},&&(L_2^{2n+1})_l^T=-K(L_2^{2n+1})_lK^{-1}\Big)\label{lulconstr}
\end{eqnarray}
Thus given the vector eigenfunction $q$, $Jq$ (or $Kq$ ) will be the
adjoint eigenfunction for the BTL (or CTL) hierarchy. This fact
connected with (\ref{bctlconstr}) and (\ref{relationadjoint}) shows
that in the BTL (or CTL) hierarchy, the adjoint case can be derived
directly from the usual case. Therefore, we can only consider the
usual case in the study of the BTL (or CTL) hierarchy.

 At last, we end this section with the introduction of the additional symmetries of the BTL and CTL
hierarchies. The Orlov-Shulman operator\cite{OS86,ASM95} is defined
as
\begin{equation}\label{osoperator}
    M\equiv(M_1,M_2)=W(\varepsilon,\varepsilon^*)W^{-1},
\end{equation}
where
$$\varepsilon=\rm{diag}[s]\Lambda^{-1},\quad \varepsilon^*=-\varepsilon^T
+\Lambda,$$
 satisfying
\begin{eqnarray}
M\Psi=(\pa_z,\pa_{z^{-1}})\Psi,&&[L,M]=(1,1), \nonumber\\
\pa_{x_{2n+1}}M=[(L^{2n+1}_1,0)_+,M],&&
\pa_{y_{2n+1}}M=[(0,L_2^{2n+1})_+,M]. \label{osproperty}
\end{eqnarray}
The additional symmetry \cite{cheng2011} can be defined by
introducing the additional independent variables $x_{m,l}^*$ and
$y_{m,l}^*$,
\begin{equation}\label{bctladdsym}
    \pa_{x_{m,l}^*}W=-(A_{1ml}(M_1,L_1),0)_-W,\quad \pa_{y_{m,l}^*}W=-(0,A_{2ml}(M_2,L_2))_-W,
\end{equation}
where $A_{iml}(M_i,L_i)$ are polynomials in $L_i$ and $M_i$. Denote
$A_{ml}({M},{L})=(A_{1ml}(M_1,L_1),A_{2ml}(M_2,L_2))$, then
\begin{itemize}
  \item in BTL case,
\begin{equation}\label{abtladdsymm}
A_{ml}({M},{L})={M}^m{L}^l-(-1)^l{L}^{l-1}{M}^m{L};
\end{equation}
  \item in CTL case,
\begin{equation}\label{actladdsymm}
A_{ml}({M},{L})={M}^m{L}^l-(-1)^l{L}^{l}{M}^m.
\end{equation}
\end{itemize}

%%%%%%%%%%%%%%%%%%%%%%%%%%%%%%%%%%%%%%%%%%%%%%%%%%%%%%
\section{The Squared Eigenfunction Symmetry for the BTL Hierarchies}
%%%%%%%%%%%%%%%%%%%%%%%%%%%%%%%%%%%%%%%%%%%%%%%%%%%%%%%
In this section, we shall construct the squared eigenfunction
symmetry for the BTL hierarchy.

Given a couple of vector eigenfunctions $q_1$ and $q_2$, \textbf{the
squared eigenfunction flow} of the BTL hierarchy can be defined by
its actions on the wave operators,
\begin{equation}\label{bgsymw}
\pa_\alpha W_1=(q_1\otimes J q_2-q_2\otimes J q_1)_l W_1,\quad
\pa_\alpha W_2=-(q_1\otimes J q_2-q_2\otimes J q_1)_u W_2,
\end{equation}
where $(A\otimes B)_{ij}=A_iB_j$ for the vectors $A$ and $B$.

 According to
(\ref{laxwaveexpression}), one can further have the squared
eigenfunction flow on the Lax operator
\begin{equation}\label{bgsymlax}
\pa_\alpha L_1=[(q_1\otimes J q_2-q_2\otimes J q_1)_l,L_1],\quad
\pa_\alpha L_2=-[(q_1\otimes J q_2-q_2\otimes J q_1)_u, L_2].
\end{equation}
Next we will show that the definitions above is well-defined:
(\ref{bgsymw}) or (\ref{bgsymlax}) is consistent with the BTL
constraint (\ref{bctlconstr}).
\begin{proposition}
$\pa_\alpha$ is consistent with the BTL constraint
(\ref{bctlconstr}), i.e. $(\pa_\alpha L_i^T)J+J(\pa_\alpha
L_i)=0,i=1,2.$
\end{proposition}
\begin{proof}
Firstly,
\begin{eqnarray*}
&&J(q_1\otimes J q_2-q_2\otimes J q_1)+(q_1\otimes J q_2-q_2\otimes J q_1)^TJ\\
&=&J(q_1\otimes q_2-q_2\otimes q_1)J+J( q_2\otimes q_1-q_1\otimes
q_2)J=0,
\end{eqnarray*}
by noting that $q_1\otimes J q_2=(q_1\otimes q_2)J^T=(q_1\otimes
q_2)J$ and $(q_1\otimes q_2)^T=q_2\otimes q_1$. Thus
\begin{eqnarray}
&&J(q_1\otimes J q_2-q_2\otimes J q_1)_l+(q_1\otimes J q_2-q_2\otimes J q_1)_l^TJ=0,\label{orthrelation1}\\
&&J(q_1\otimes J q_2-q_2\otimes J q_1)_u+(q_1\otimes J
q_2-q_2\otimes J q_1)_u^TJ=0,\label{orthrelation2}
\end{eqnarray}
from the fact if $JA+A^TJ=0$, then $JA_l+A_l^TJ=0$ and
$JA_u+A_u^TJ=0$ (see \cite{uenotaksasai}).

Then for $i=1$, from (\ref{bctlconstr})(\ref{bgsymlax}),
\begin{eqnarray*}
&&(\pa_\alpha L_1^T)J+J(\pa_\alpha L_1)\\
&=&[(q_1\otimes J q_2-q_2\otimes J q_1)_l,L_1]^TJ+J[(q_1\otimes J q_2-q_2\otimes J q_1)_l,L_1]\\
&=&-[(q_1\otimes J q_2-q_2\otimes J q_1)_l^T,L_1^T]J+J[(q_1\otimes J q_2-q_2\otimes J q_1)_l,L_1]\\
&=&(q_1\otimes J q_2-q_2\otimes J q_1)_l^TJL_1-JL_1J^{-1}(q_1\otimes J q_2-q_2\otimes J q_1)_l^TJ+J[(q_1\otimes J q_2-q_2\otimes J q_1)_l,L_1]\\
&=&-J(q_1\otimes J q_2-q_2\otimes J q_1)_lL_1+JL_1(q_1\otimes J q_2-q_2\otimes J q_1)_l+J[(q_1\otimes J q_2-q_2\otimes J q_1)_l,L_1]\\
&=&-J[(q_1\otimes J q_2-q_2\otimes J q_1)_l,L_1]+J[(q_1\otimes J
q_2-q_2\otimes J q_1)_l,L_1]=0.
\end{eqnarray*}
The case $i=2$ can be similarly proved.
\end{proof}
Thus $\pa_\alpha$ is indeed well-defined. we next will show that
this squared eigenfunction flow is indeed a kind of symmetry for the
BTL hierarchy, and thus is called \textbf{the squared eigenfunction
symmetry}.
\begin{proposition}
\begin{eqnarray}
[\pa_\alpha,\pa_{x_{2n+1}}]=[\pa_\alpha,\pa_{y_{2n+1}}]=0.
\end{eqnarray}
\end{proposition}
\begin{proof}
In fact, according to (\ref{bctlhierarchy}), (\ref{wevolution}),
(\ref{bgsymw}) and (\ref{bgsymlax})
\begin{eqnarray*}
&&[\pa_\alpha,\pa_{x_{2n+1}}]W_1\\
&=&\pa_\alpha\Big((L_1^{2n+1})_uW_1\Big)-\pa_{x_{2n+1}}\Big((q_1\otimes J q_2-q_2\otimes Jq_1)_lW_1\Big)\\
&=&[(q_1\otimes J q_2-q_2\otimes Jq_1)_l,L_1^{2n+1}]_uW_1+(L_1^{2n+1})_u(q_1\otimes J q_2-q_2\otimes Jq_1)_lW_1\\
&&-((L_1^{2n+1})_uq_1\otimes J q_2)_lW_1+(q_1\otimes(L_1^{2n+1})_u^TJ q_2)_lW_1+((L_1^{2n+1})_uq_2\otimes J q_1)_lW_1\\
&&-(q_2\otimes (L_1^{2n+1})_u^TJ q_1)_lW_1-(q_1\otimes Jq_2-q_2\otimes Jq_1)_l(L_1^{2n+1})_uW_1\\
&=&[(q_1\otimes J q_2-q_2\otimes Jq_1)_l,L_1^{2n+1}]_uW_1+[(L_1^{2n+1})_u,(q_1\otimes J q_2-q_2\otimes Jq_1)_l]W_1\\
&&+[(q_1\otimes J q_2-q_2\otimes Jq_1),(L_1^{2n+1})_u]_lW_1\\
&=&[(q_1\otimes J q_2-q_2\otimes Jq_1)_l,(L_1^{2n+1})_u]_uW_1+[(L_1^{2n+1})_u,(q_1\otimes J q_2-q_2\otimes Jq_1)_l]W_1\\
&&+[(q_1\otimes J q_2-q_2\otimes Jq_1)_l,(L_1^{2n+1})_u]_lW_1\\
&=&[(q_1\otimes J q_2-q_2\otimes
Jq_1)_l,(L_1^{2n+1})_u]W_1+[(L_1^{2n+1})_u,(q_1\otimes J
q_2-q_2\otimes Jq_1)_l]W_1=0.
\end{eqnarray*}
Note that $q_1\otimes(L_1^{2n+1})_u^TJ q_2=q_1\otimes
Jq_2(L_1^{2n+1})_u$ is used in the third identity. While
$[A_u,B_u]_l=[A_l,B_l]_u=0$ is used in the fourth identity.

Similarly,
$[\pa_\alpha,\pa_{x_{2n+1}}]W_2=[\pa_\alpha,\pa_{y_{2n+1}}]W_1=[\pa_\alpha,\pa_{y_{2n+1}}]W_2=0$
can be proved.
\end{proof}
Define the following double expansions
\begin{eqnarray*}
Y_1(\lambda,\mu)&=&\sum_{m=0}^\infty
\frac{(\mu-\lambda)^m}{m!}\sum_{l=-\infty}^\infty\lambda^{-l-m-1}A_{1m,m+l}(M_1,L_1),\\
Y_2(\lambda,\mu)&=&\sum_{m=0}^\infty
\frac{(\mu-\lambda)^m}{m!}\sum_{l=-\infty}^\infty\lambda^{-l-m-1}A_{2m,m+l}(M_2,L_2),
\end{eqnarray*}
which can be viewed as the generator of the additional symmetries
for the BTL hierarchy. This double expansions can be related with
the wave functions in the following way \cite{cheng2011} by
considering (\ref{relationadjoint}),
\begin{lemma}\label{propaddisymmtodab5}
For the BTL hierarchy,
\begin{eqnarray}
{Y}_1(\lambda,\mu)&=&\lambda^{-1}({\Psi}_1(x,y;\mu)\otimes J{\Psi}_1(x,y;-\lambda)-{\Psi}_1(x,y;-\lambda)\otimes J{\Psi}_1(x,y;\mu)),\label{generatortodab1}\\
{Y}_2(\lambda,\mu) &=&\lambda^{-1}({\Psi}_2(x,y;\mu^{-1})\otimes
J{\Psi}_2(x,y;-\lambda^{-1})-{\Psi_2}(x,y;-\lambda^{-1})\otimes
J{\Psi}_2(x,y;\mu^{-1})).\label{generatortodab2}
\end{eqnarray}
\end{lemma}
In \cite{cheng2011}, there are some mistakes in the corresponding
results about $Y(\lambda,\mu)$ for BTL and CTL (see Proposition 7
and 13 in \cite{cheng2011}), and the correct ones should be without
``$(\quad)_-$".

We denote the squared eigenfunction symmetry generated by
$\Psi_1(x,y;\mu)$ and $\lambda^{-1}\Psi_1(x,y;-\lambda)$ as
$\pa_{\alpha_1}$, while the one generated by $-\Psi_2(x,y;\mu^{-1})$
and $\lambda^{-1}\Psi_2(x,y;-\lambda^{-1})$ as $\pa_{\alpha_2}$.
Then
\begin{eqnarray}
\pa_{\alpha_1} W_1&=&\lambda^{-1}({\Psi}_1(x,y;\mu)\otimes
J{\Psi}_1(x,y;-\lambda)-{\Psi}_1(x,y;-\lambda)\otimes
J{\Psi}_1(x,y;\mu))_l
W_1,\\
\pa_{\alpha_1} W_2&=&-\lambda^{-1}({\Psi}_1(x,y;\mu)\otimes
J{\Psi}_1(x,y;-\lambda)-{\Psi}_1(x,y;-\lambda)\otimes
J{\Psi}_1(x,y;\mu))_u W_2,
\end{eqnarray}
and
\begin{eqnarray}
\pa_{\alpha_2} W_1&=&-\lambda^{-1}({\Psi}_2(x,y;\mu^{-1})\otimes
J{\Psi}_2(x,y;-\lambda^{-1})-{\Psi_2}(x,y;-\lambda^{-1})\otimes
J{\Psi}_2(x,y;\mu^{-1}))_l
W_1,\\
\pa_{\alpha_2} W_2&=&\lambda^{-1}({\Psi}_2(x,y;\mu^{-1})\otimes
J{\Psi}_2(x,y;-\lambda^{-1})-{\Psi_2}(x,y;-\lambda^{-1})\otimes
J{\Psi}_2(x,y;\mu^{-1}))_u W_2.
\end{eqnarray}

Further from (\ref{bctladdsym}) (\ref{generatortodab1}) and
(\ref{generatortodab2}), we have

\begin{proposition}\label{btlsqsadd}
The squared eigenfunction symmetries $\pa_{\alpha_1}$ and
$\pa_{\alpha_2}$ are the generators of the additional symmetries for
the BTL hierarchy, that is,
\begin{eqnarray}
\pa_{\alpha_1}&=&\sum_{m=0}^\infty
\frac{(\mu-\lambda)^m}{m!}\sum_{k=-\infty}^\infty\lambda^{-k-m-1}\pa_{x_{m,m+k}^*},\label{sqaddbtl1}\\
\pa_{\alpha_2}&=&\sum_{m=0}^\infty
\frac{(\mu-\lambda)^m}{m!}\sum_{k=-\infty}^\infty\lambda^{-k-m-1}\pa_{y_{m,m+k}^*}.\label{sqaddbtl2}
\end{eqnarray}
\end{proposition}

Thus we have establish the relation between the squared
eigenfunction symmetry and the additional symmetry.
%%%%%%%%%%%%%%%%%%%%%%%%%%%%%%%%%%%%%%%%%%%%%%%%%%%%%%
\section{The Squared Eigenfunction Symmetry for the CTL Hierarchies}
%%%%%%%%%%%%%%%%%%%%%%%%%%%%%%%%%%%%%%%%%%%%%%%%%%%%%%%
In this section, the squared eigenfunction symmetry for the CTL
hierarchy will be given.

Similar to the case of the above section, given two eigenfunctions
$q_1$ and $q_2$, one can define \textbf{the squared eigenfunction
flow} of the CTL hierarchy by its actions on the wave operators,
\begin{equation}\label{cgsymw}
\pa_\alpha W_1=(q_1\otimes K q_2+q_2\otimes K q_1)_l W_1,\quad
\pa_\alpha W_2=-(q_1\otimes K q_2+q_2\otimes K q_1)_u W_2.
\end{equation}

 According to
(\ref{laxwaveexpression}), the action of the squared eigenfunction
flow on the Lax operator is
\begin{equation}\label{cgsymlax}
\pa_\alpha L_1=[(q_1\otimes K q_2+q_2\otimes K q_1)_l,L_1],\quad
\pa_\alpha L_2=-[(q_1\otimes K q_2+q_2\otimes K q_1)_u, L_2].
\end{equation}
The next proposition shows that the definitions above is
well-defined.
\begin{proposition}
$\pa_\alpha$ is consistent with the CTL constraint
(\ref{bctlconstr}), i.e. $(\pa_\alpha L_i^T)K+K(\pa_\alpha
L_i)=0,i=1,2.$
\end{proposition}
\begin{proof}
\begin{eqnarray*}
&&K(q_1\otimes K q_2+q_2\otimes K q_1)+(q_1\otimes K q_2+q_2\otimes K q_1)^TK\\
&=&K(q_1\otimes q_2+q_2\otimes q_1)K^T+K( q_1\otimes q_2+q_2\otimes
q_1)K=0,
\end{eqnarray*}
by noting that $K^T=-K$. Then
\begin{eqnarray}
&&K(q_1\otimes K q_2+q_2\otimes K q_1)+(q_1\otimes K q_2+q_2\otimes K q_1)^TK=0,\label{sprelation1}\\
&&K(q_1\otimes K q_2+q_2\otimes K q_1)+(q_1\otimes K q_2+q_2\otimes
K q_1)^TK=0,\label{sprelation2}
\end{eqnarray}
from the fact if $KA+A^TK=0$, then $KA_l+A_l^TK=0$ and
$KA_u+A_u^TK=0$ (see \cite{uenotaksasai}).

The rest of the proof is similarly to the case of the BTL hierarchy.
\end{proof}
Thus $\pa_\alpha$ is indeed well-defined. By the same way as the BTL
case, one can get the following proposition, which shows that this
squared eigenfunction flow is indeed a kind of symmetry for the CTL
hierarchy, and thus is called \textbf{the squared eigenfunction
symmetry}.
\begin{proposition}
\begin{eqnarray}
[\pa_\alpha,\pa_{x_{2n+1}}]=[\pa_\alpha,\pa_{y_{2n+1}}]=0.
\end{eqnarray}
\end{proposition}

Define the generator of the additional symmetries for the CTL
hierarchy as the following double expansions
\begin{eqnarray*}
Y_1(\lambda,\mu)&=&\sum_{m=0}^\infty
\frac{(\mu-\lambda)^m}{m!}\sum_{l=-\infty}^\infty\lambda^{-l-m-1}A_{1m,m+l}(M_1,L_1),\\
Y_2(\lambda,\mu)&=&\sum_{m=0}^\infty
\frac{(\mu-\lambda)^m}{m!}\sum_{l=-\infty}^\infty\lambda^{-l-m-1}A_{2m,m+l}(M_2,L_2).
\end{eqnarray*}
This double expansions can be related with the wave functions in the
following way\cite{cheng2011} by considering
(\ref{relationadjoint}),
\begin{lemma}\label{propaddisymmtodac5}
For the CTL hierarchy,
\begin{eqnarray}
{Y}_1(\lambda,\mu)&=&{\Psi}_1(x,y;\mu)\otimes K{\Psi}_1(x,y;-\lambda)+{\Psi}_1(x,y;-\lambda)\otimes K {\Psi}_1(x,y;\mu),\label{generatortodac1}\\
{Y}_2(\lambda,\mu)&=&\lambda^{-2}{\Psi}_2(x,y;\mu^{-1})\otimes
K{\Psi}_2(x,y;-\lambda^{-1})+\mu^{-2}{\Psi}_2(x,y;-\lambda^{-1})\otimes
K{\Psi}_2(x,y;\mu^{-1}).\label{generatortodac2}
\end{eqnarray}
\end{lemma}
If denote the squared eigenfunction symmetry generated by
$\Psi_1(x,y;\mu)$ and $\Psi_1(x,y;-\lambda)$ as $\pa_{\alpha_1}$,
then
\begin{eqnarray}
\pa_{\alpha_1} W_1&=&({\Psi}_1(x,y;\mu)\otimes
K{\Psi}_1(x,y;-\lambda)+{\Psi}_1(x,y;-\lambda)\otimes K
{\Psi}_1(x,y;\mu))_l
W_1,\\
\pa_{\alpha_1} W_2&=&-({\Psi}_1(x,y;\mu)\otimes
K{\Psi}_1(x,y;-\lambda)+{\Psi}_1(x,y;-\lambda)\otimes K
{\Psi}_1(x,y;\mu))_u W_2,
\end{eqnarray}
And denote the squared eigenfunction symmetry generated by
$-\lambda^{-1}\Psi_2(x,y;\lambda^{-1})$ and
$\lambda^{-1}\Psi_2(x,y;-\lambda^{-1})$ as $\pa_{\alpha_2}$, that is
\begin{eqnarray}
\pa_{\alpha_2}W_1&=&-\lambda^{-2}({\Psi}_2(x,y;\mu^{-1})\otimes
K{\Psi}_2(x,y;-\lambda^{-1})+{\Psi_2}(x,y;-\lambda^{-1})\otimes
K{\Psi}_2(x,y;\mu^{-1}))_lW_1,\label{paalphactl21}\\
\pa_{\alpha_2} W_2&=&\lambda^{-2}({\Psi}_2(x,y;\mu^{-1})\otimes
K{\Psi}_2(x,y;-\lambda^{-1})+{\Psi_2}(x,y;-\lambda^{-1})\otimes
K{\Psi}_2(x,y;\mu^{-1}))_u W_2.\label{paalphactl22}
\end{eqnarray}
From (\ref{bctladdsym}), (\ref{generatortodac1}) and
(\ref{generatortodac2}), we have,
\begin{proposition}\label{ctlsqsadd}
The squared eigenfunction symmetries $\pa_{\alpha_1}$ are the
generators of the additional symmetries $\pa_{x_{ml}^*}$ for the CTL
hierarchy, that is,
\begin{eqnarray}
\pa_{\alpha_1}&=&\sum_{m=0}^\infty
\frac{(\mu-\lambda)^m}{m!}\sum_{k=-\infty}^\infty\lambda^{-k-m-1}\pa_{x_{m,m+k}^*}.\label{sqaddctl1}
\end{eqnarray}
And the relation between the squared eigenfunction symmetry
$\pa_{\alpha_2}$ and the additional symmetries $\pa_{y_{ml}^*}$ for
the CTL hierarchy is as follows:
\begin{eqnarray}
\pa_{\alpha_2}&=&\sum_{k=-\infty}^\infty\lambda^{-k-1}\pa_{y_{0,k}^*}.\label{sqaddctl2}
\end{eqnarray}
\end{proposition}
\noindent{\bf Remark:} Usually, one may think that the result about
the CTL hierarchy should be parallel to those about the BTL
hierarchy. But here (\ref{sqaddctl2}) is different from
(\ref{sqaddbtl2}). Because of the coefficients of
(\ref{generatortodac2}), it is difficult to construct the squared
eigenfunction symmetry corresponding to the generating function of
the additional symmetries $\pa_{y_{ml}^*}$. But when $\mu=\lambda$,
we can construct $\pa_{\alpha_2}$ as (\ref{paalphactl21}) and
(\ref{paalphactl22}). In this case, $\pa_{\alpha_2}$ can be viewed
as the generating function of $\pa_{y_{0,k}^*}$. In order to get the
parallel result to the BTL hierarchy, some modifications of
(\ref{cgsymw}) and (\ref{cgsymlax}) may be needed.

%%%%%%%%%%%%%%%%%%%%%%%%%%%%%%%%%%%%%%%%%%%%%%%%%%%%%%%%%%%
\section{Conclusions and Discussions}
%%%%%%%%%%%%%%%%%%%%%%%%%%%%%%%%%%%%%%%%%%%%%%%%%%%%%%%%%%%
The squared eigenfunction symmetries for the BTL and CTL hierarchies
are constructed explicitly (see (\ref{bgsymw}), (\ref{bgsymlax}),
(\ref{cgsymw}) and (\ref{cgsymlax})) in the suitable combination of
the ones of the TL hierarchy by considering the BTL and CTL
constraint. And the relation with the additional symmetry is also
investigated, that is, the squared eigenfunction symmetry can be
viewed as the generating function of the additional symmetries when
the defined eigenfunctions are the wave functions (see Proposition
\ref{btlsqsadd} and \ref{ctlsqsadd}). These theories are expected to
be applied in the study of the symmetry constraints for the BTL and
CTL hierarchies and the corresponding additional symmetries.

{\bf Acknowledgments}

{\noindent \small This work is supported by the NSFC (Grant No. 11226196) and ``the Fundamental Research Funds for the
Central Universities" No. 2012QNA45}

\end{document}